\documentclass[aps,prapplied,twocolumn,superscriptaddress]{revtex4-2}
\usepackage{graphicx}
\usepackage{url}
\graphicspath{{./Figures/}}
\usepackage{amsmath,amsthm,amssymb,physics,color,enumerate,tabularx,makecell,mathrsfs,float,bbm,stackrel,mathtools,multirow,bm}
\usepackage{natbib}
\usepackage{amsmath}
\usepackage[colorlinks=true,%
bookmarks=false,%
linkcolor=blue,%
urlcolor=blue,%
citecolor=blue,%
breaklinks]{hyperref}
\usepackage{soul} %to overstrike words, \st{bla} 

\usepackage[dvipsnames]{xcolor}

\newtheorem{prop}{Proposition}

\newtheorem{corollary}{Corollary}

\begin{document}

\title{Daemonic ergotropy in continuously-monitored open quantum batteries}

\author{Daniele Morrone}
\email{daniele.morrone@unimi.it}
\affiliation{Quantum Technology Lab, Dipartimento di Fisica {\it Aldo Pontremoli}, Universit\`{a} degli Studi di Milano, I-20133 Milano, Italy}
\affiliation{Algorithmiq Ltd., Kanavakatu 3C, FI-00160 Helsinki, Finland}

\author{Matteo A. C. Rossi}%
\affiliation{Algorithmiq Ltd., Kanavakatu 3C, FI-00160 Helsinki, Finland}
\affiliation{InstituteQ - the Finnish Quantum Institute, Aalto University, Finland}
\affiliation{QTF Centre of Excellence, Department of Applied Physics, Aalto University, FI-00076 Aalto, Finland}

\author{Marco G. Genoni}
\email{marco.genoni@fisica.unimi.it}
\affiliation{Quantum Technology Lab, Dipartimento di Fisica {\it Aldo Pontremoli}, Universit\`{a} degli Studi di Milano, I-20133 Milano, Italy}

\date{\today}% It is always \today, today,
	
\begin{abstract}
The amount of work that can be extracted from a quantum system can be increased by exploiting the information obtained from a measurement performed on a correlated ancillary system. The concept of daemonic ergotropy has been introduced to properly describe and quantify this work extraction enhancement in the quantum regime. We here explore the application of this idea in the context of continuously-monitored open quantum systems, where information is gained by measuring the environment interacting with the energy-storing quantum device. 
We first show that the corresponding daemonic ergotropy takes values between the ergotropy and the energy of the corresponding unconditional state. The upper bound is achieved by assuming an initial pure state and a perfectly efficient projective measurement on the environment, independently of the kind of measurement performed.
On the other hand, if the measurement is inefficient or the initial state is mixed, the daemonic ergotropy is generally dependent on the measurement strategy. This scenario is investigated via a paradigmatic example of an open quantum battery: a two-level atom driven by a classical field and whose spontaneously emitted photons are continuously monitored via either homodyne, heterodyne, or photo-detection.
\end{abstract}

\maketitle
	
%\tableofcontents
	
\section{Introduction}
The field of quantum thermodynamics aims to extend the laws of classical thermodynamics to the quantum regime~\cite{binderThermodynamicsQuantumRegime2018}. One of its main goals is to understand the limits of work extraction from a quantum system, which has both fundamental and practical implications. Allahverdyan et al. \cite{allahverdyanMaximalWorkExtraction2004} introduced the concept of ergotropy, as the maximum amount of work that can be extracted from a quantum state through unitary dynamics. 

Extracting work from a quantum state is extremely useful in the context of quantum batteries (QBs), energy-storing devices that follow the laws of quantum mechanics during charging and discharging processes.
Research on QBs has then focused on exploring quantum enhancements in the charging process \cite{alickiEntanglementBoostExtractable2013,binderQuantacellPowerfulCharging2015,campaioliEnhancingChargingPower2017,campaioliQuantumBatteriesReview2018,ferraroHighPowerCollectiveCharging2018,andolina_extractable_2019,ZhangPowerfulHarmonicCharging2019,CrescenteChargingEnergyFluctuations2020,JuliaFarreBounds2020,rossiniQuantumAdvantageCharging2020,gyhmQuantumChargingAdvantage2022,seahQuantumSpeedupCollisional2021,salviaQuantumAdvantageCharging2022,landiBatteryChargingCollision2021,mayoCollectiveEffectsQuantum2022,barraEfficiencyFluctuationsQuantum2022,rodriguezOptimalQuantumControl2022,QiMagnonMediatedQuantum2021,CrescenteEnhancingCoherentEnergy2022,mazzonciniOptimalControlMethods2023,rodriguezAIdiscoveryNewCharging2023} and various implementations have been proposed, including single qubits, collective spins, and quantum harmonic oscillators \cite{campaioliQuantumBatteriesReview2018,andolinaChargermediatedEnergyTransfer2018,rosaUltrastableChargingFastscrambling2020,rossiniQuantumAdvantageCharging2020,shaghaghiMicromasersQuantumBatteries2022}. Recent experiments have shown promising results in realizing a quantum-enhanced QB \cite{huOptimalChargingSuperconducting2022,quachSuperabsorptionOrganicMicrocavity2022,joshiExperimentalInvestigationQuantum2022,wennigerExperimentalAnalysisEnergy2023}. 

The analysis of open quantum batteries (OQBs), that is QBs in an open quantum system scenario, is crucial to discuss and guarantee their real-world implementation.
Research on OQBs has focused so far on studying the effect of different environmental models~\cite{farinaChargermediatedEnergyTransfer2019,zakavatiBoundsChargingPower2021,carregaDissipativeDynamicsOpen2020,kaminNonMarkovianEffectsCharging2020,morroneChargingQuantumBattery2023} and devising quantum control strategies to counteract the impact of noise~\cite{santosStableAdiabaticQuantum2019,quachUsingDarkStates2020,gherardiniStabilizingOpenQuantum2020,mitchisonChargingQuantumBattery2021,rodriguezOptimalQuantumControl2022,yaoOptimalChargingOpen2022a,koshiharaQuantumErgotropyQuantum2023}. 
We here want to address the situation where some information leaking to the environment can be recovered via continuous measurements~\cite{wisemanQuantumMeasurementControl2009,jacobsQuantumMeasurementTheory2014}, and then exploited for enhancing the work extraction protocol (see Fig.~\ref{f:OQB}). Since the paradigmatic example of Maxwell's demon, it is well known that acquiring information via a measurement effectively brings a system out of equilibrium and allows to extract useful work via conditional operations~\cite{maruyamaColloquiumPhysicsMaxwell2009}. This idea was then brought to the quantum realm, also highlighting the relationship between extractable work and quantum correlations~\cite{lloydQuantummechanicalMaxwellDemon1997,oppenheimThermodynamicalApproachQuantifying2002,zurekQuantumDiscordMaxwell2003,maruyamaThermodynamicalDetectionEntanglement2005,kimQuantumSzilardEngine2011,perarnau-llobetExtractableWorkCorrelations2015,brownPassivityPracticalWork2016,francicaDaemonicErgotropyEnhanced2017,ciampiniExperimentalExtractableWorkbased2017,brunelliDetectingGaussianEntanglement2017,bernardsDaemonicErgotropyGeneralised2019,safranekWorkExtractionUnknown2023}. Francica et al. introduced the concept of {\it daemonic ergotropy}~\cite{francicaDaemonicErgotropyEnhanced2017}, as the maximum average work that can be extracted from a quantum system via unitary operations by exploiting the information obtained by measuring a correlated ancilla.
 \begin{figure}[t!]
	\centerline{\includegraphics[width=\columnwidth]{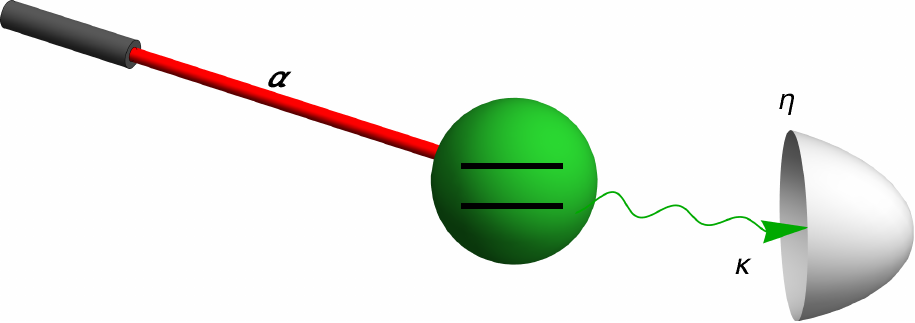}}
	\caption{Pictorial representation of a continuously-monitored quantum battery: a coherent drive supplies energy with intensity $\alpha$ to a two-level atom (green) which is spontaneously emitting into its environment with rate $\kappa$. The emitted field is continuously monitored via a detector with efficiency $\eta$.}
    \label{f:OQB}
 \end{figure}

Continuously-monitored open quantum systems~\cite{wisemanQuantumMeasurementControl2009,jacobsQuantumMeasurementTheory2014}  have been extensively studied from a theoretical point of view, mainly for feedback-assisted quantum state engineering protocols~\cite{wisemanQuantumTheoryOptical1993,wisemanQuantumTheoryContinuous1994,dohertyFeedbackControlQuantum1999,thomsenSpinSqueezingQuantum2002,serafiniDeterminationMaximalGaussian2010,szorkovszkyMechanicalSqueezingParametric2011,genoniOptimalFeedbackControl2013,genoniQuantumCoolingSqueezing2015,hoferEntanglementenhancedTimecontinuousQuantum2015,martinDeterministicGenerationRemote2015,brunelliConditionalDynamicsOptomechanical2019,martinQuantumFeedbackMeasurement2020,jiangOptimalityFeedbackControl2020,zhangLocallyOptimalMeasurementbased2020,digiovanniUnconditionalMechanicalSqueezing2021,candeloroFeedbackAssistedQuantumSearch2023,isaksenMechanicalCoolingSqueezing2023} and for quantum estimation purposes~\cite{mabuchiDynamicalIdentificationOpen1996,gambettaStateDynamicalParameter2001,geremiaQuantumKalmanFiltering2003,gutaOptimalEstimationQubit2008,tsangOptimalWaveformEstimation2010,tsangFundamentalQuantumLimit2011,tsangQuantumMetrologyOpen2013,gammelmarkBayesianParameterInference2013,gammelmarkFisherInformationQuantum2014,sixParameterEstimationMeasurements2015,catanaFisherInformationsLocal2015,kiilerichBayesianParameterEstimation2016,ralphMultiparameterEstimationQuantum2017,genoniCramErRaoBound2017,albarelliUltimateLimitsQuantum2017,rossiExperimentalAssessmentEntropy2020,iliasCriticalityEnhancedQuantumSensing2022,yangEfficientInformationRetrieval2022}; more recently their quantum thermodynamics properties have been also analyzed, being a paradigmatic example of out-of-equilibrium quantum systems~\cite{elouardWorkHeatEntropy2018,manzanoQuantumThermodynamicsContinuous2022,garrahanThermodynamicsQuantumJump2010,garrahanQuantumTrajectoryPhase2011,garrahanCatchingReversingQuantum2018,cilluffoQuantumJumpStatistics2019,belenchiaEntropyProductionContinuously2020,rossiExperimentalAssessmentEntropy2020,landiIrreversibleEntropyProduction2021,landiInformationalSteadyStates2022,bhandariContinuousMeasurementBoosted2022,yanikThermodynamicsQuantumMeasurement2022,bhandariMeasurementBasedQuantumThermal2023}. Remarkably, the possibility of observing single trajectories has been now experimentally shown in different platforms, such as superconducting circuits~\cite{murchObservingSingleQuantum2013,campagne-ibarcqObservingQuantumState2016,ficheuxDynamicsQubitSimultaneously2018,minevCatchReverseQuantum2019}, optomechanical~\cite{wieczorekOptimalStateEstimation2015,rossiObservingVerifyingQuantum2019} and hybrid~\cite{thomasEntanglementDistantMacroscopic2021} systems. Recently, feedback protocols able to cool mechanical oscillators have also been demonstrated~\cite{rossiMeasurementbasedQuantumControl2018,magriniRealtimeOptimalQuantum2021,tebbenjohannsQuantumControlNanoparticle2021}.

From a fundamental point of view, during its evolution, an open quantum system is correlated to the environment; for this reason, being able to measure the environment makes this scenario the ideal playground where to investigate the properties of the daemonic ergotropy both for its fundamental aspects, but also for its possible practical implementation in OQBs. 

In this paper we derive the general properties of the daemonic ergotropy in open quantum systems and we apply them to the simplest example of an OQB, that is a two-level atom driven via a classical field and spontaneously emitting photons into its electromagnetic environment. We show that the daemonic ergotropy of a continuously-monitored system surpasses the ergotropy of the unconditional state and, with perfectly efficient measurements, can even reach the energy of the unconditional state. We then discuss the performance of different types of measurements with non-unit efficiency.

The manuscript is organized as follows: in Sec. \ref{s:DE} we review the concept of daemonic ergotropy and present our first results. In Sec. \ref{s:oqs} we extend the concept of daemonic ergotropy to continuously-monitored open quantum systems, along with some general results that apply in this scenario. In Sec. \ref{s:tls} we discuss these in a minimal example of an open quantum battery: a two level system driven by a classical field and whose spontaneous emitted photons are continuously monitored. Finally, in Sec. \ref{s:conclusion}, we give our conclusions and propose further outlooks.

\section{Daemonic ergotropy}
\label{s:DE}
We start by recalling the definition and main properties of the ergotropy of a quantum state. Let us consider a quantum system described by a Hamiltonian $\hat{H}_0$, where, without losing generality, we fix its smallest eigenvalue equal to zero. We define the ergotropy $\mathcal{E}(\varrho)$ of a given quantum state $\varrho$ describing such quantum system as the maximum amount of work that can be extracted via unitary dynamics~\cite{allahverdyanMaximalWorkExtraction2004}:
\begin{align}
    \mathcal{E}(\varrho) = \max_{\hat{U}} \left[ E(\varrho) -  E(\hat{U} \varrho \hat{U}^\dag)\right] \,,
\end{align}
where $E(\varrho)= \Tr[\hat{H}_0\varrho]$ denotes the average energy of the quantum state $\varrho$. A closed formula for $\mathcal{E}(\varrho)$ in terms of eigenvalues and eigenvectors of $\varrho$ and $\hat{H}_0$ can be straightforwardly derived and the following main properties can be demonstrated: i) the ergotropy is upper bounded by the energy, and this upper bound is saturated for pure quantum states, i.e $\mathcal{E}(|\psi\rangle\langle\psi|)=E(|\psi\rangle\langle\psi|)$; ii) ergotropy is equal to zero $\mathcal{E}(\varrho)=0$ if and only if $\varrho$ is a {\it passive state}, i.e. it is diagonal in the Hamiltonian $\hat{H}_0$ eigenbasis, and its eigenvalues do not admit energy inversion~\cite{allahverdyanMaximalWorkExtraction2004}. 

Let us now consider a bipartite quantum state $\varrho^{\tiny SA}$, where the quantum system $S$ represents our energy storing device, while $A$ is an ancillary system. If the ancilla is discarded, the maximum amount of extractable work is simply equal to $\mathcal{E}(\varrho^{\tiny S})$, with $\varrho^{\tiny S} = \Tr_A[\varrho^{\tiny SA}]$. We now assume that a measurement is performed on the ancilla, described by a positive operator-valued measurement (POVM) $\{\hat{\Pi}_a^A\}$. We can thus define the {\it daemonic ergotropy} as the average ergotropy of the corresponding conditional states~\cite{francicaDaemonicErgotropyEnhanced2017}
\begin{align}
\overline{\mathcal{E}}_{\{\hat{\Pi}_a^A\}} = \sum_a p_a \mathcal{E}(\varrho_a^{\tiny S})\,,
\end{align}
where
\begin{align}
    p_a &= \Tr_{SA}[\varrho^{\tiny SA} (\hat{\mathbbm{1}}^S \otimes \hat{\Pi}_a^A)]\\
    \varrho_a^{\tiny S} &= \Tr_A[\varrho^{\tiny SA} (\hat{\mathbbm{1}}^S \otimes \hat{\Pi}_a^A)]/p_a
\end{align}

denote respectively the probability and the conditional state corresponding to the measurement outcome $a$. The ergotropy is a convex quantity in the quantum state $\varrho$~\cite{francicaDaemonicErgotropyEnhanced2017,bernardsDaemonicErgotropyGeneralised2019}; as a consequence, since $\varrho^{\tiny S}=\sum_a p_a \varrho_a^{\tiny S}$, one obtains that
\begin{align}
\overline{\mathcal{E}}_{\{\hat{\Pi}_a^A\}}  \geq \mathcal{E}(\varrho^{\tiny S}) \,, \label{eq:daemonconvex}
\end{align}
that is the average work that may be extracted is increased thanks to the information obtained from the ancilla. Furthermore, the daemonic ergotropy can be rewritten as
\begin{align}
\overline{\mathcal{E}}_{\{\hat{\Pi}_a^A\}} = E(\varrho^{\tiny S}) - \sum_a p_a \min_{\hat{U}_a} E(\hat{U}_a \varrho_a^{\tiny S} \hat{U}_a^\dag)\,.
\label{eq:daemontwo}
\end{align}
From this formula it is then clear that, in order to achieve the daemonic enhancement, one needs to implement a conditional unitary evolution $\hat{U}_a$ that depends on the conditional state $\varrho_a^{\tiny S}$. Furthermore, $\overline{\mathcal{E}}_{\{\hat{\Pi}_a^A\}}$ generally depends on the specific POVM $\{\hat{\Pi}_a^A\}$ implemented (see~\cite{bernardsDaemonicErgotropyGeneralised2019} for a recipe to obtain the optimal POVM  maximizing $\overline{\mathcal{E}}_{\{\hat{\Pi}_a^A\}}$). Before addressing the scenario of continuosly-monitored quantum systems, we here present the first result of our work via the following proposition.

\begin{prop}\label{prop1}
Given a bipartite system system+ancilla prepared in an initial pure state $|\Psi\rangle^{\tiny SA}$, and assuming a projective (rank-one) measurement on the ancilla $\{\Pi_a^{A}=|\phi_a\rangle\langle \phi_a|\}$, then 
\begin{align}
\overline{\mathcal{E}}_{\{\hat{\Pi}_a^A\}} =E(\varrho^{\tiny S}) \,, 
\end{align}
that is, the daemonic ergotropy is equal to the energy of the reduced state of the system $\varrho^{\tiny S} = \Tr_A[\varrho^{\tiny SA}]$, independently of the measurement performed.
\end{prop}
\begin{proof}
In order to prove this theorem, we first observe that the conditional state remains pure for any measurement outcome, i.e. $\varrho_a^{\tiny S}=|\xi_a\rangle\langle \xi_a|$, with
\begin{align}
    |\xi_a\rangle=\frac{1}{\sqrt{p_a}} \langle \phi_a | \Psi\rangle^{\tiny SA}
\end{align}
As a consequence one has 
\begin{align}
\overline{\mathcal{E}}_{\{\hat{\Pi}_a^A\}} &= \sum_a p_a \mathcal{E}(|\xi_a\rangle\langle \xi_a|) \nonumber \\
&= \sum_a p_a E(|\xi_a\rangle\langle \xi_a|) = E(\varrho^{\tiny S}) \,, 
\end{align}
where we have exploited: i) the fact that the ergotropy of pure states is equal to their energy, ii) the linearity of the trace, and iii) the relationship $\varrho^{\tiny S} =  \sum_a p_a |\xi_a\rangle\langle \xi_a|$.
\end{proof}

\section{Daemonic ergotropy in continuously-monitored open quantum systems} \label{s:oqs} We start by briefly introducing the basic notions on continuously-monitored quantum systems and quantum trajectories~\cite{wisemanQuantumMeasurementControl2009,jacobsQuantumMeasurementTheory2014}. We assume that our quantum system is interacting with a Markovian environment such that its unconditional evolution is described by a master equation in Lindblad form
\begin{align}
    %\frac{d\varrho_{\sf unc}(t)}{dt} 
    d\varrho_{\sf unc}(t)/dt = -i [\hat{H}_s(t),\varrho_{\sf unc}(t)] + \mathcal{D}[\hat{c}]\varrho_{\sf unc}(t) \,,  
    \label{eq:masterequation}
\end{align}
where $\hat{H}_s(t)$ denotes the Hamiltonian ruling the evolution of the system and $\mathcal{D}[\hat{c}]\varrho= \hat{c}\varrho \hat{c} - (\hat{c}^\dag \hat{c} \varrho + \varrho\hat{c}^\dag \hat{c})/2$ is the Lindbladian superoperator~\footnote{We are here considering a single jump operator $\hat{c}$ describing the interaction between the system and a zero-temperature environment, but the formalism can be straightforwardly generalized to multiple jump operators and to generic thermal environments}.

If one assumes that the environment is continuously monitored, the system evolution will be described by a stochastic master equation (SME) for the conditional state $\varrho_{\sf c}(t)$, which is typically referred to as a quantum trajectory. In general it is always true that
\begin{align}
    \varrho_{\sf unc}(t) = \sum_{\sf traj} p_{\sf traj} \, \varrho_{\sf c}(t), \label{eq:convex}
\end{align}

where $p_{\sf traj}$ denotes the probability of each quantum trajectory $\varrho_{\sf c}(t)$. Different measurement strategies will correspond to different unravelling of the unconditional master equation (\ref{eq:masterequation}), that is to different SMEs for the conditional state $\varrho_{\sf c}(t)$ and, mathematically speaking, to different convex combinations for the unconditional state $\varrho_{\sf unc}(t)$. We will later consider the most paradigmatic examples of such unravellings, corresponding to the scenarios where the environment is continuously-monitored via either photo-detection (PD), homodyne-detection (HoD) or heterodyne-detection (HeD). 
The explicit formulas for the corresponding SMEs can be found in appendix \ref{a:monitoring}. Remarkably, while we here focus on Markovian open quantum systems described by a Lindblad master equation (\ref{eq:masterequation}), all the results that follow apply whenever one can write the unconditional state as a mixture of trajectories as of (\ref{eq:convex}), including monitored quantum systems exhibiting non-Markovian behaviour.

Besides the kind of detection performed on the environment, such unravellings are characterized by their measurement efficiency $\eta$, that comprehensively quantifies the portion of the environment that is accessible and the efficiency of the detector. In particular, we recall that for $\eta=1$, that is when one assumes that the environment is fully accessible and measurable via a projective (rank-one) measurement, and for an initial pure state for the system, one can prove that the evolution can be described via a stochastic Schrödinger equation, and the conditional state remains pure during the whole dynamics~\cite{wisemanQuantumMeasurementControl2009,jacobsQuantumMeasurementTheory2014}. We can now present one of the main results via the following proposition:

\begin{prop}\label{prop2}
Given an open quantum system, whose charging process is described by a (possibly time-dependent) Hamiltonian $\hat{H}_s(t)$, and whose environment can be continuously monitored with efficiency $\eta$, the corresponding daemonic ergotropy 
\begin{align}
    \overline{\mathcal{E}}_{\sf unr,\eta}(t) = \sum_a p_{\sf traj} \mathcal{E}(\varrho_{\sf c}(t))
\end{align}
is bounded as
\begin{align}
\mathcal{E}(\varrho_{\sf unc}(t)) \leq \overline{\mathcal{E}}_{\sf unr,\eta}(t) \leq {E}(\varrho_{\sf unc}(t)) \,. \label{eq:daemonunravel}
\end{align}
The upper bound can be achieved in the presence of Markovian environment, whenever the system is initially prepared in a pure state and the monitoring is performed with unit efficiency $\eta=1$, independently of the kind of unravelling, i.e. 
\begin{align}
\overline{\mathcal{E}}_{\sf unr,\eta=1}(t) = {E}(\varrho_{\sf unc}(t)). \, .
\end{align}
\end{prop}
\begin{proof}
The first inequality in Eq.~(\ref{eq:daemonunravel}) is a generalization of Eq. (\ref{eq:daemonconvex}), applied to (\ref{eq:convex}). Similarly the upper bound is a consequence of Eq. (\ref{eq:daemontwo}), while the fact that the upper bound can be achieved for unravellings of pure states follows straightforwardly from Prop.~\ref{prop1}.
\end{proof}
In the following, we show that as a corollary, the inequality above can be trivially extended to the case where ergotropies and energies are rescaled by the evolution time $t$, and thus in terms of figure of merits characterizing the charging power of the protocol.

\subsection{Average power of continuously-monitored quantum batteries}
\label{a:corollary}
Let us assume that the state of the quantum battery at the beginning of the charging process is initially in a pure state ($\varrho_0= |\psi_0\rangle\langle\psi_0|$) , and a charging protocol evolve the system from a time $t=0$ to $t=\tau$. We can use the definitions of energy ($E$), ergotropy ($\mathcal{E}$) and daemonic ergotropy ($\overline{\mathcal{E}}_{\{\hat{\Pi}_a^A\}}$) introduced in the manuscript, to define respectively the average power ($P$), the average ergotropic power ($\mathcal{P}$) and the average daemonic power ($\overline{\mathcal{P}}_{\{\hat{\Pi}_a^A\}}$) as follows:
\begin{align}
    P( \varrho(\tau))&=\frac{E(\varrho (\tau))-E(\varrho_0)}{\tau}\label{eq:pwr} \\
    \mathcal{P}( \varrho(\tau))&=\frac{\mathcal{E}(\varrho (\tau))-E(\varrho_0)}{\tau} \label{eq:erg_pwr}\\
    \overline{\mathcal{P}}_{\{\hat{\Pi}_a^A\}}(\tau) &= \frac{ \overline{\mathcal{E}}_{\{\hat{\Pi}_a^A\}}(\varrho(\tau)) - E(\varrho_0)}{\tau} 
    \label{eq:daem_pwr}
\end{align}
We can thus formulate the following corollary, that can be easily proven starting from Prop. 2 of the manuscript.

\begin{corollary}\label{corollary2}
Given an open quantum system, whose charging process is described by a (possibly time-dependent) Hamiltonian $\hat{H}_s(t)$, and whose environment can be continuously monitored with efficiency $\eta$, the corresponding average daemonic power $\overline{\mathcal{P}}_{\sf unr,\eta}(t) = \sum_a p_{\sf traj} \mathcal{P}(\varrho_{\sf c}(t))$ is bounded as
\begin{align}
\mathcal{P}(\varrho_{\sf unc}(t)) \leq \overline{\mathcal{P}}_{\sf unr,\eta}(t) \leq {P}(\varrho_{\sf unc}(t)) \,. \label{eq:daemonunravelpower}
\end{align}
The upper bound can be achieved in the presence of Markovian environment, whenever the system is initially prepared in a pure state and the monitoring is performed with unit efficiency $\eta=1$, independently of the kind of unravelling, i.e. 
\begin{align}
\overline{\mathcal{P}}_{\sf unr,\eta=1}(t) = {P}(\varrho_{\sf unc}(t)) \, .
\end{align}
\end{corollary}
 
We have thus proved that, as expected, the extractable work and power can be increased in an open quantum system if one obtains some information by monitoring the environment. Remarkably, the maximum daemonic ergotropy is equal to the unconditional energy and can be achieved via unit efficiency monitoring, independently of the measurement strategy. In the following we will rather investigate what happens in the more practical and experimentally relevant scenario of monitoring with non-unit efficiency, where a hierarchy between the different unravellings is established.
\begin{figure}[t]
	\includegraphics[width=0.5\textwidth]{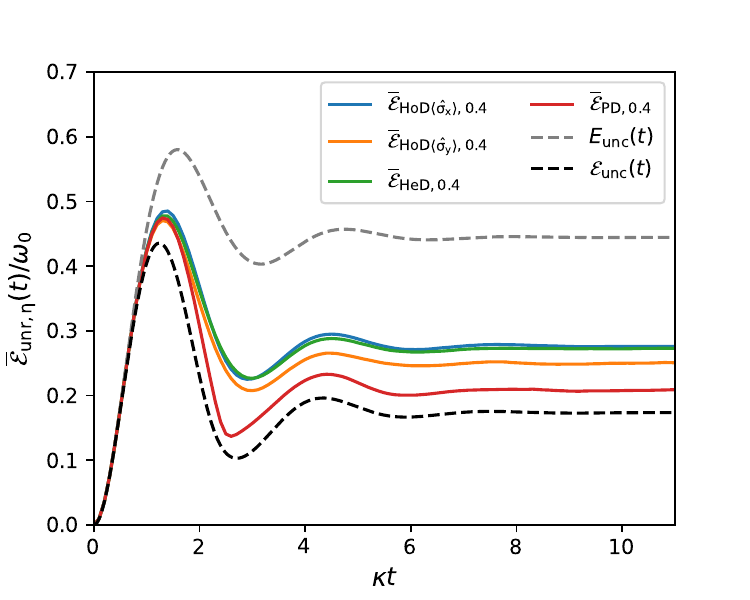}
	\caption{Daemonic ergotropies $\overline{\mathcal{E}}_{\sf unr,\eta}(t)$ as a function of time for different unravellings (HoD with photocurrents monitoring either $\hat{\sigma_x}$ or $\hat{\sigma}_y$, HeD and PD) with $\alpha/\kappa=1,\,\eta=0.4$ and averaged over $n=5 \times 10^4$ trajectories. Black and gray dashed lines correspond respectively to the ergotropy and energy of the unconditional state. The different daemonic ergotropies lie  as expected between these two lines, with HeD and $\hat{\sigma}_x$-HoD giving the larger values of extractable work. The upper and the lower bounds would be saturated for all the unravellings in the case of respectively perfect monitoring ($\eta=1$) and no-monitoring ($\eta=0$).}
	\label{f:time1}
\end{figure}

\section{A continuously-monitored open quantum battery}\label{s:tls} We now consider the paradigmatic example of an OQB, that is a two-level atom characterized by a Hamiltonian $\hat{H}_0 = (\omega_0/2) (\hat{\sigma}_z+1)$, driven by a resonant classical field of intensity $\alpha$, which acts as a charger, and spontaneously emitting with rate $\kappa$ (see Fig. \ref{f:OQB}). In interaction picture with respect to $\hat{H}_0$, the evolution of the system is described by a Markovian master equation of the form (\ref{eq:masterequation}), where the Hamiltonian ruling the evolution and the jump operator respectively read $\hat{H}_s = \alpha \hat{\sigma}_x$ and $\hat{c}=\sqrt{\kappa}\sigma_-$, i.e.
\begin{align}
    \frac{d\varrho_{\sf unc}(t)}{dt} 
    %d\varrho_{\sf unc}(t)/dt 
    = -i \alpha [\hat{\sigma}_x,\varrho_{\sf unc}(t)] + \kappa \mathcal{D}[\hat{\sigma}_-]\varrho_{\sf unc}(t) \,.
    \label{eq:OQBme}
\end{align}

An analytical solution can be obtained for $\varrho_{\sf unc}(t)$ and we report here the corresponding steady-state values of energy and ergotropy
\begin{align}
    E_{\sf unc}^{\sf ss}/\omega_0=&\frac{4\alpha^2}{8\alpha^2+\kappa} \,, \label{eq:energyss}
     \\
    %\mathcal{E}_{\sf ss}/\omega_0=&\frac{1}{2}\left(\sqrt{1-\frac{64\alpha^4}{(8\alpha^2+\kappa^2)^2}} - \frac{\kappa^2}{8\alpha^2+\kappa^2} \right) \,,\\
    %&= \frac{\sqrt{16 \alpha^2 \kappa^2 + \kappa^4 }-\kappa^2}{2(8\alpha^2+\kappa^2)} \\
     \mathcal{E}_{\sf unc}^{\sf ss}/\omega_0=& \frac{\kappa}{2} \left( \frac{\sqrt{16 \alpha^2 + \kappa^2 }-\kappa}{8\alpha^2+\kappa^2} \right),
    \label{eq:ergotropyss}
\end{align}
where we have exploited the following formula for the ergotropy of a qubit state
\begin{align}
    \mathcal{E}(\varrho)=E(\varrho)+\frac{\omega_0}{2}\sqrt{2\mu (\varrho)-1},
    \label{eq:qubitergotropy}
\end{align}
in  terms of energy and purity $\mu(\varrho) = \Tr[\varrho^2]$. One observes that the steady state energy $E_{\sf unc}^{\sf ss}$  grows monotonically as a function of $\alpha/\kappa$ and asymptotically reaches the value ${E}_{\sf unc,max}^{\sf ss} = \omega_0/2$ in the limit of large driving. Oppositely, the steady-state ergotropy $\mathcal{E}_{\sf unc}^{\sf ss}$ presents a maximum at 
\begin{align}
\alpha/\kappa =\sqrt{(1+\sqrt{2})/8} \, ,
\end{align} 
where it reaches its peak value 
\begin{align}
\mathcal{E}_{\sf unc,max}^{\sf ss}= \omega_0(\sqrt{2}-1)/2 \, .
\end{align}

If we now assume that a photo-counting detector is able to measure the spontaneously emitted photons with efficiency $\eta$,  the dynamics of the conditional states $\varrho_{\sf c}(t)$ is described by a SME of the form (\ref{eq:SMEPD}); in particular the statistics of the corresponding Poissonian increment is univocally identified by its average value
\begin{align}
\mathbbm{E}[dN_t] = \eta\kappa \langle \hat{\sigma}_+ \hat{\sigma}_-\rangle_t dt \, , 
\end{align}
and thus depends on the average value of $\hat{\sigma}_z$ (we remind that $\hat{\sigma}_+ \hat{\sigma}_-=(\hat{\sigma}_z + \hat{\mathbbm{1}}_2)/2$).

If one rather considers a homodyne detection on the emitted field, one has a SME of the form (\ref{eq:SMEHoD}), where in particular the continuous homodyne photo-current reads
\begin{align}
dy_t &= \sqrt{\eta \kappa} \langle \hat{\sigma}_- e^{i \phi} + \hat{\sigma}_+ e^{-i \phi} \rangle_t \,dt + dW_t \,,  \\
&= \sqrt{\eta \kappa} \langle \cos\phi \hat{\sigma}_x  + \sin\phi  \hat{\sigma}_y  \rangle_t \,dt + dW_t  \,,
\end{align}
where $\phi$ corresponds to the homodyne phase. In particular for $\phi=0$ and $\phi=\pi/2$ one has photocurrents depending respectively to the average values of $\hat{\sigma}_x$ and $\hat{\sigma}_y$.

Heterodyne detection on the environment leads to a similar SME, as can be indeed thought and implemented as a double homodyne, measuring orthogonal quadratures with half efficiency. As a consequence one gets a SME of the form (\ref{eq:SMEHeD}) where the two photocurrents depend on $\hat{\sigma}_x$ and $\hat{\sigma}_y$ as
\begin{align}
dy_t^{(1)} &= \sqrt{\frac{\eta \kappa}{2}} \langle \hat{\sigma}_x  \rangle_t \,dt + dW_t^{(1)} \,,  \\
dy_t^{(2)}  &= \sqrt{\frac{\eta \kappa}{2}} \langle   \hat{\sigma}_y  \rangle_t \,dt + dW_t^{(2)}  \,.
\end{align}

We will denote with $\overline{\mathcal{E}}_{\sf PD,\eta},\;\overline{\mathcal{E}}_{\sf HoD,\eta}$ and $\overline{\mathcal{E}}_{\sf HeD,\eta}$ the daemonic ergotropies corresponding to continuous monitoring of the fluorescence field due to the atomic spontaneous emission via respectively PD, HoD and HeD with efficiency $\eta$. 

According to Prop~\ref{prop2}, we know that for unravellings of pure states with $\eta=1$, one obtains $\overline{\mathcal{E}}_{\sf unr,\eta=1}(t) = E(\varrho_{\sf unc}(t))$, that is, the energy of the unconditional state can be fully extracted via conditional unitary operations for all the possible detection strategies. 

We first consider the situation where the system is initially prepared in the ground state $|0\rangle$, while the environment is not fully accessible and it is thus monitored with non-unit efficiency $\eta$. We have numerically solved the corresponding SMEs~\cite{johanssonQuTiPOpensourcePython2012,johanssonQuTiPPythonFramework2013} and evaluated the corresponding daemonic ergotropy by averaging over a large number of trajectories. In this case we find that there is a hierarchy between the different unravellings: as one can observe in Fig. \ref{f:time1}, for $\eta=0.4$ and $\alpha=\kappa$, homodyne detection with photocurrent proportional to $\langle \hat{\sigma}_x\rangle_t$ (corresponding to $\phi=0$) and heterodyne detection lead to the largest values of daemonic ergotropy, while photo-detection is the least performing strategy. The results presented in this plot are quite general: in all our numerical simulations we find that, as regards HoD, $\overline{\mathcal{E}}_{\sf HoD,\eta=1}(t)$ is maximized (minimized) for $\phi=0\;(\phi=\pi/2)$, that is for a $\hat{\sigma}_x$-dependent ($\hat{\sigma}_y$-dependent) photocurrent.
Furthermore, in the whole range of parameters we have explored, $\hat{\sigma}_x$-HoD and HeD generally yield very similar values of daemonic ergotropy. The same behavior can be indeed observed if we focus on the steady-state properties as shown in Fig.~\ref{f:alpha1}, where steady-state daemonic ergotropies $\overline{\mathcal{E}}_{\sf unr,\eta}^{\sf ss}$ for the different unravellings are plotted as a function of $\alpha$ and for two different values of $\eta$. From the figure one can clearly observe the hierarchy existing between the different strategies, and how, by increasing the monitoring efficiency, the daemonic ergotropy can reach values, approaching the unconditional energy $E_{\sf unc}^{\sf ss}$ in the limit of $\eta=1$. 
 \begin{figure}
	\centerline{\includegraphics[width=0.5\textwidth]{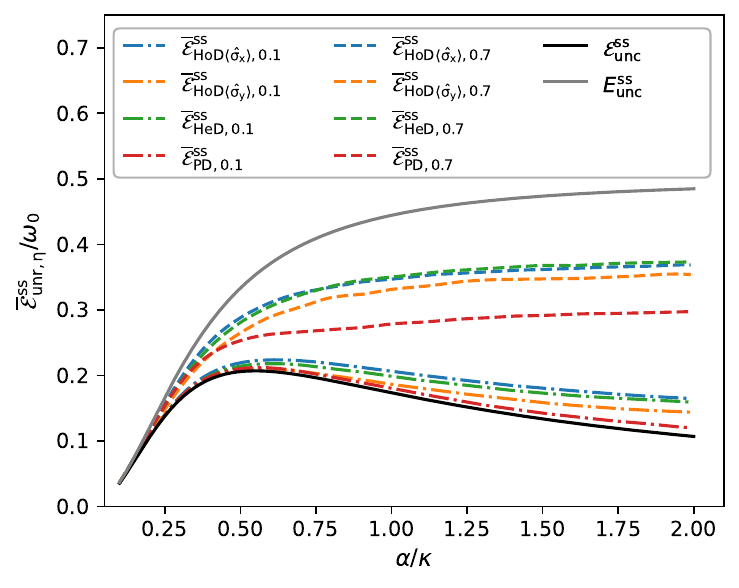}}
	\caption{Steady-state daemonic ergotropy $\overline{\mathcal{E}}_{\sf unr,\eta}^{\sf ss}$  for the different unravellings (HoD with photocurrents monitoring either $\hat{\sigma}_x$ or $\hat{\sigma}_y$, HeD, PD) as a function of $\alpha$ and for $\eta=0.1$ (dashed-dotted) and $\eta=0.7$ (dashed) (number of trajectories $n=5\times 10^4$). Also in this case HeD and $\hat{\sigma}_x$-HoD yield the largest values. Black and gray solid lines correspond respectively to ergotropy and energy of the unconditional steady-state, that is to the lower and upper bound for the daemonic ergotropies, that would be obtained for respectively no-monitoring ($\eta=0$) or perfect monitoring ($\eta=1$).\\}
    \label{f:alpha1}
 \end{figure}
\begin{figure}[t]
	\centerline{\includegraphics[width=0.5\textwidth]{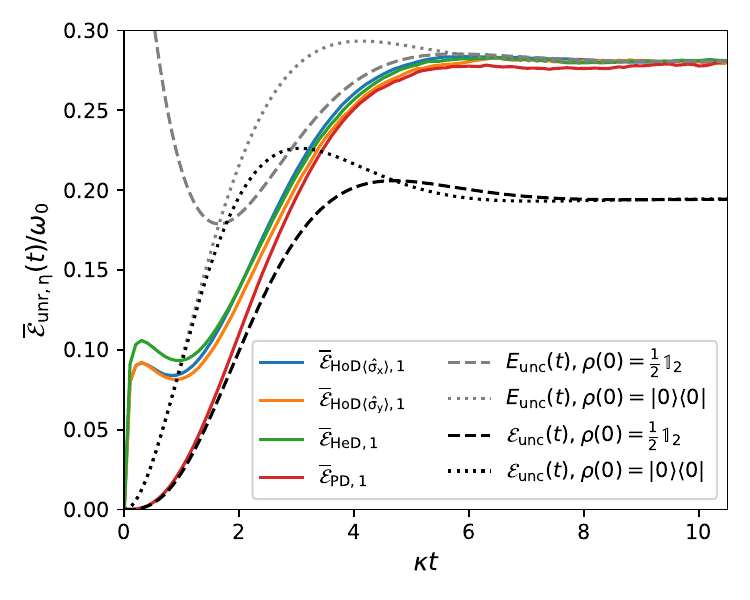}} 
	\caption{Daemonic ergotropy $\overline{\mathcal{E}}_{\sf unr,\eta}(t)$ as a function of time for different unravellings (HoD with photocurrents monitoring either $\hat{\sigma_x}$ or $\hat{\sigma}_y$, HeD, PD) with $\alpha/\kappa=0.4$, $\eta=1$ and by considering the maximally mixed state $\rho_0 = \hat{\mathbbm{1}}/2$ as initial state (number of trajectories $n=5 \times 10^4$). 
 Black and gray lines correspond respectively to ergotropy and energy of the unconditional state when initial states are the maximally mixed state (dashed) and the ground state (dotted).  } 
 
	\label{f:time2}
\end{figure}
Finally, in Fig.~\ref{f:time2} we consider the other scenario where different unravellings lead to different daemonic ergotropies, that is when the environment is monitored with unit efficiency, but the initial state is mixed. In particular, we consider the maximally mixed state $\varrho_0 = \hat{\mathbbm{1}}/2$ as initial state and we plot $\overline{\mathcal{E}}_{\sf unr,\eta}(t)$ as a function of time. We observe that at steady state all the unravellings lead to the same value of daemonic ergotropy,  equal to the corresponding unconditional energy: the monitoring will eventually purify all the trajectories and thus one falls back into the scenario described in Prop.~\ref{prop2}. However, different unravellings lead to a different purification speed, an effect that has been widely discussed in the literature~\cite{jacobsHowProjectQubits2003,combesRapidStateReduction2006,wisemanReconsideringRapidQubit2006,jordanQubitFeedbackControl2006,wisemanOptimalityFeedbackControl2008,ruskovQubitPurificationSpeedup2012}. With respect to the scenarios described in these works, we here have a non-trivial system Hamiltonian $\hat{H}_s$ ruling the dynamics, and a fixed jump operator $\hat{c}=\sqrt{\kappa}\hat{\sigma}_-$, representing the interaction between system and environment. Remarkably, at small times we find that the purification speed of HeD and, in the second instance, HoD (almost independently on the phase $\phi$) allows to achieve values of daemonic ergotropy evidently larger than the maximum daemonic ergotropy obtainable by starting from a ground state, whose upper bound (the unconditional energy represented as a dotted-gray line in Fig.~\ref{f:time2}) can be interpreted as the amount of energy injected into the system by the driving laser; we can thus conclude that in monitoring-enhanced battery charging protocols, the effects of purification may be more efficient than pure energy injection at small times. 
On the other hand we also observe that PD still yields the lowest values of daemonic ergotropy in the transient leading to steady-state, and for small times the enhancement due to the monitoring respect to the unconditional ergotropy is almost negligible.

\section{Conclusions} \label{s:conclusion} We have extended the concept of daemonic ergotropy to the open quantum system scenario, where some information leaking into the environment can be continuously monitored and exploited in order to enhance the work extraction protocol. Our findings reveal that the daemonic ergotropy not only surpasses the unconditional state ergotropy, but can even reach its energy in the ideal scenario of unit efficiency detectors. We have then discussed the simplest, but practically relevant, example of an OQB, that is a two-level atom classically driven by an external field acting as a charger. 

Our main results (Props.\ref{prop1},\ref{prop2}) require very few assumptions compared to the ones implied in the model of OQB considered. In particular, they hold even for charging protocols with time-dependent Hamiltonian~\cite{ZhangPowerfulHarmonicCharging2019,CrescenteChargingEnergyFluctuations2020}, allowing one to further assess the performances of these protocols in the case of continuous monitoring and to envisage more efficient quantum feedback protocols. Moreover, one can also consider trajectories describing continuously monitored open quantum systems in the presence of a non-Markovian environment; while the interpretation of unravellings of non-Markovian master equation in terms of monitoring is in general not guaranteed~\cite{DiosiNonMarkov2008,WisemanNonMarkov2008,PiiloNonMarkov2008,BarchielliNonMarkov2010,SmirneNonMarkov2020}, our approach can be directly pursued in the non-Markovian setting by describing such quantum conditional dynamics via continuously measured collisional models~\cite{ciccarelloQuantumCollisionModels2022,morroneChargingQuantumBattery2023,landiInformationalSteadyStates2022}.

A proof-of-principle experimental demonstration of our results can be readily pursued in a circuit-QED platform, where quantum trajectories corresponding to HeD of atom fluorescence have been recently observed~\cite{campagne-ibarcqObservingQuantumState2016,ficheuxDynamicsQubitSimultaneously2018}. In general our results pave the way to further investigation on the relationship between measurement energy cost and work extraction in continuously-monitored quantum systems~\cite{alickiThermodynamicsQuantumInformation2004,sagawaSecondLawThermodynamics2008,sagawaMinimalEnergyCost2009,jacobsSecondLawThermodynamics2009,vaccaroInformationErasureEnergy2011,jacobsQuantumMeasurementFirst2012,deffnerQuantumWorkThermodynamic2016,manzanoOptimalWorkExtraction2018,abdelkhalekFundamentalEnergyCost2018}, and on the design of a new generation of monitoring-enhanced and noise-resilient quantum batteries.

\begin{acknowledgments} The authors acknowledge helpful discussions with F. Albarelli, A. Serafini and A. Smirne. DM acknowledges financial support from MUR under
the “PON Ricerca e Innovazione 2014-2020”.
MACR acknowledges financial support from the Academy of Finland via the Centre of Excellence program (Project No. 336810).
MGG acknowledges support from UniMi via PSR-2 2021.
The authors wish to acknowledge CSC – IT Center for Science, Finland, for computational resources.
\end{acknowledgments}
\bibliography{reference}
\appendix
\onecolumngrid

%
%5
\section{\label{a:monitoring} Continuously-monitored quantum systems}
In this appendix we briefly introduce the formalism behind continuously monitored quantum systems, presenting the stochastic master equations (SMEs) corresponding to continuous photo-detection, homodyne-detection and heterodyne-detection. \newline

We assume that our quantum system is interacting with a Markovian environment such that its unconditional evolution is described by a master equation in Lindblad form
\begin{align}
    \frac{d\varrho_{\sf unc}(t)}{dt} 
    %d\varrho_{\sf unc}(t)/dt 
    = -i [\hat{H}_s,\varrho_{\sf unc}(t)] + \mathcal{D}[\hat{c}]\varrho_{\sf unc}(t) \,,  
    \label{eq:masterequationSM}
\end{align}
where $\hat{H}_s$ denotes the Hamiltonian ruling the evolution of the system and $\mathcal{D}[\hat{c}]\varrho= \hat{c}\varrho \hat{c} - (\hat{c}^\dag \hat{c} \varrho + \varrho\hat{c}^\dag \hat{c})/2$ is the Lindbladian superoperator. We remark that we will here consider a single jump operator $\hat{c}$ describing the interaction between the system and a zero-temperature environment, but the formalism can be straightforwardly generalized to multiple jump operators and to generic thermal environments. 

We will now assess the scenario where the environment is continuously monitored. Different measurement strategies lead to different possible evolutions of the corresponding conditional states $\varrho_{\sf c}(t)$, i.e. to different unravellings of the unconditional master equation (\ref{eq:masterequationSM}). We will here consider three different measurements: photo-detection (PD), homodyne-detection (HoD) and heterodyne-detection (HeD). 

In the case of PD the evolution is described by a stochastic master equation 
\begin{align}
d\varrho_{\sf c}(t) &= -i [\hat{H}_s,\varrho_{\sf c}(t)]\,dt + (1-\eta)\mathcal{D}[\hat{c}]\varrho_{\sf c}(t)\,dt \nonumber  \\
&\,\,\ - \frac{\eta}{2} (\hat{c}^\dag \hat{c} \varrho_{\sf c}(t)+ \varrho_{\sf c}(t)\hat{c}^\dag \hat{c}) + \eta \langle \hat{c}^\dag \hat{c}\rangle_t \varrho_c(t) \, dt \nonumber \\
&\,\,\, +\left( \frac{\hat{c}\varrho_{\sf c}(t)\hat{c}^\dag}{\langle \hat{c}^\dag \hat{c}\rangle_t} - \varrho_{\sf c}(t) \right) \,dN_t \,, \label{eq:SMEPD}
\end{align}
where $\langle \hat{A}\rangle_t = \Tr[\varrho_{\sf c}(t) \hat{A}]$, $\eta$ is the efficiency of the detector and $dN_t$ is a Poisson increment taking value $0$ (no-click event) or $1$ (detector click event), and having average value $\mathbbm{E}[dN_t] = \eta \langle \hat{c}^\dag \hat{c}\rangle_t dt$.

For HoD one has the following diffusive SME
\begin{align}
    d\varrho_{\sf c}(t) &= -i [\hat{H}_s,\varrho_{\sf c}(t)]\,dt + \mathcal{D}[\hat{c}]\varrho_{\sf c}(t)\,dt \nonumber \\
    &\,\,\, + \sqrt{\eta} \mathcal{H}[\hat{c}e^{i\phi}]\varrho_{\sf c} \,dW_t \, \label{eq:SMEHoD}
\end{align}
where $\mathcal{H}[\hat{c}]\varrho = \hat{c}\varrho + \varrho\hat{c}^\dag  -\langle \hat{c} + \hat{c}^\dag\rangle_t \varrho$, $\phi$ is the phase of the quadrature monitored via the homodyne, and $dW_t$ is a Wiener increment related to the measured photocurrent $dy_t = \sqrt{\eta}\langle \hat{c}e^{i\phi} + \hat{c}^\dag e^{-i\phi}\rangle_t \,dt  + dW_t$. The unravelling for HeD is a generalization of the HoD case, corresponding to a double-homodyne scheme leading to the SME
\begin{align}
    d\varrho_{\sf c}(t) &= -i [\hat{H}_s,\varrho_{\sf c}(t)]\,dt + \mathcal{D}[\hat{c}]\varrho_{\sf c}(t)\,dt \nonumber \\
    &\,\,\, + \sqrt{\eta/2}\, \mathcal{H}[\hat{c}]\varrho_{\sf c} \,dW_t^{(1)} \nonumber \\ 
    &\,\,\, + \sqrt{\eta/2}\, \mathcal{H}[i\hat{c}]\varrho_{\sf c} \,dW_t^{(2)}\, \label{eq:SMEHeD}
\end{align}
where $dW_t^{(1)}$ and $dW_t^{(2)}$ are uncorrelated Wiener increments corresponding to the two photocurrents $dy_t^{(1)}=\sqrt{\eta/2}\,\langle \hat{c} + \hat{c}^\dag \rangle_t \,dt  + dW_t^{(1)}$ and $dy_t^{(2)}=\sqrt{\eta/2}\,\langle i\hat{c} - i\hat{c}^\dag \rangle_t \,dt  + dW_t^{(2)}$.

\end{document}